\documentclass[11pt]{article}
\usepackage{epsf}
\usepackage{amsmath}
\usepackage{epsfig}
\usepackage{times}
\usepackage{amssymb}
\usepackage{amsthm}
\usepackage{setspace}
\usepackage{cite}

\usepackage{algorithmic}  
\usepackage{algorithm}

\usepackage{shadow}
\usepackage{fancybox}
\usepackage{fancyhdr}

\def\w{{\bf w}}

\def\y{{\bf y}}

\def\x{{\bf x}}

\def\x{{\mathbf x}}

\def\w{{\bf w}}

\def\x{{\bf x}}
\def\y{{\bf y}}
\def\z{{\bf z}}

\def\h{{\bf h}}

\def\be{\begin{equation}}
\def\ee{\end{equation}}
\def\ba{\left[\begin{array}}
\def\ea{\end{array}\right]}

\def\t{{\bf t}}

\def\w{{\bf w}}

\def\x{{\bf x}}
\def\y{{\bf y}}
\def\z{{\bf z}}

\def\1{{\bf 1}}

\def\0{{\bf 0}}

\def\erfinv{\mbox{erfinv}}
\def\htheta{\hat{\theta}}

\def\hw{\bar{\h}}
\def\hwh{\tilde{\h}}

\def\cweak{c_{w}}

\def\Swp{S_w^{(p)}}
\def\Swh{S_w^{(h)}}

\newtheorem{theorem}{Theorem}

\newtheorem{lemma}{Lemma}

\begin{document}

\begin{singlespace}

\title {Towards a better compressed sensing 
\footnote{ This work was supported in part by NSF grant \#CCF-1217857.}
}
\author{
\textsc{Mihailo Stojnic}
\\
\\
{School of Industrial Engineering}\\
{Purdue University, West Lafayette, IN 47907} \\
{e-mail: {\tt mstojnic@purdue.edu}} }
\date{}
\maketitle

\centerline{{\bf Abstract}} \vspace*{0.1in}

In this paper we look at a well known linear inverse problem that is one of the mathematical cornerstones of the compressed sensing field. In seminal works \cite{CRT,DOnoho06CS} $\ell_1$ optimization and its success when used for recovering sparse solutions of linear inverse problems was considered. Moreover, \cite{CRT,DOnoho06CS} established for the first time in a statistical context that an unknown vector of linear sparsity can be recovered as a known existing solution of an under-determined linear system through $\ell_1$ optimization. In \cite{DonohoPol,DonohoUnsigned} (and later in \cite{StojnicCSetam09,StojnicUpper10}) the precise values of the linear proportionality were established as well. While the typical $\ell_1$ optimization behavior has been  essentially settled through the work of \cite{DonohoPol,DonohoUnsigned,StojnicCSetam09,StojnicUpper10}, we in this paper look at possible upgrades of $\ell_1$ optimization. Namely, we look at a couple of algorithms that turn out to be capable of recovering a substantially higher sparsity than the $\ell_1$. However, these algorithms assume a bit of ``feedback" to be able to work at full strength. This in turn then translates the original problem of improving upon $\ell_1$ to designing algorithms that would be able to provide output needed to feed the $\ell_1$ upgrades considered in this papers.

\vspace*{0.25in} \noindent {\bf Index Terms: Compressed sensing; $\ell_1$ optimization; linear systems of equations;
$\ell_1$-optimization}.

\end{singlespace}

\section{Introduction}
\label{sec:back}

We start by looking at the mathematical description of the linear inverse problems of interest in this paper. Namely, these problems will essentially be under-determined systems of linear equations that are known to have sparse solutions. These problems are one of the mathematical cornerstones of a very popular compressed sensing field (of course a great deal of work has been done in the compressed sensing; instead of reviewing it here we for more on compressed sensing ideas refer to the introductory papers \cite{CRT,DOnoho06CS}). As such they are consequently one of the subjects of consideration in almost any of the papers related to compressed sensing. A series of our own recent work \cite{StojnicCSetam09,StojnicICASSP09,StojnicUpper10} is of course no an exception. What is typically intriguing about these problems is the simplicity of their statements.

To insure that we are on a right mathematical track we will along these lines start with providing their an as simple as possible description. One typically starts with a systems matrix $A$ which is an $m\times n$ ($m\leq n$) dimensional matrix with real entries and then considers an $n$ dimensional vector $\tilde{\x}$ that also has real entries but on top of that no more than $k$ nonzero entries (in the rest of this paper we will call such a vector $k$-sparse). Then one forms the product of $A$ and $\tilde{\x}$ to obtain $\y$
\begin{equation}
\y=A\tilde{\x}. \label{eq:defy}
\end{equation}
Clearly, in general $\y$ is an $m$ dimensional vector with real entries. Then, for a moment one pretends that $\tilde{\x}$ is not known and poses the following inverse problem: given $A$ and $\y$ from (\ref{eq:defy}) can one then determine $\tilde{\x}$? Or in other words, can one for a given pair $A$ and $\y$ find the $k$ sparse solution of the following linear systems of equation type of problem (see, Figure \ref{fig:model})
\begin{equation}
A\x=\y. \label{eq:system}
\end{equation}
\begin{figure}[htb]
\centering
\centerline{\epsfig{figure=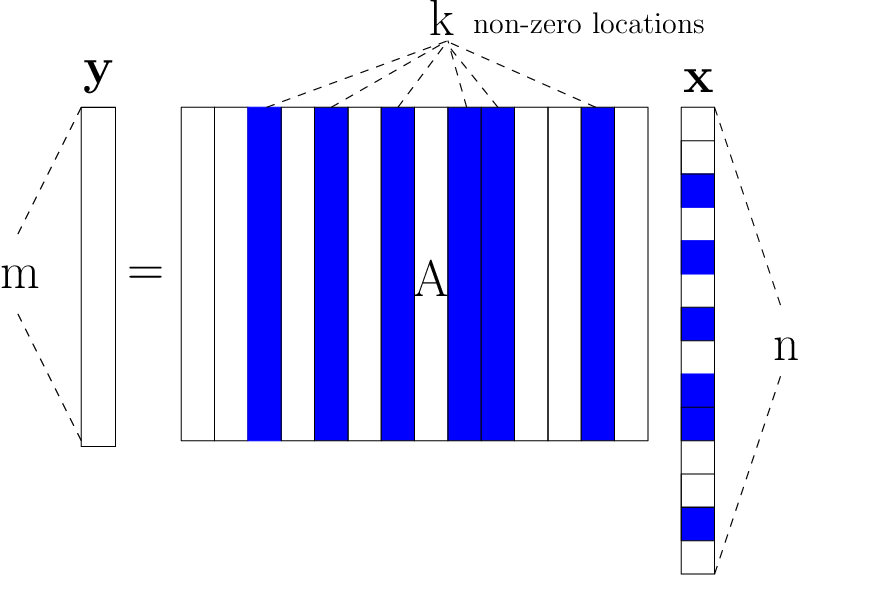,width=10.5cm,height=6cm}}
\caption{Model of a linear system; vector $\x$ is $k$-sparse}
\label{fig:model}
\end{figure}
Of course, based on (\ref{eq:defy}) such an $\x$ exists (moreover, it is an easy algebraic exercise to show that when $k<m/2$ it is in fact unique). Additionally, we will assume that there is no $\x$ in (\ref{eq:system}) that is less than $k$ sparse. One often (especially within the compressed sensing context) rewrites the problem described above (and given in (\ref{eq:system})) in the following way
\begin{eqnarray}
\mbox{min} & & \|\x\|_{0}\nonumber \\
\mbox{subject to} & & A\x=\y, \label{eq:l0}
\end{eqnarray}
where $\|\x\|_{0}$ is what is typically called $\ell_0$ norm of vector $\x$. For all practical purposes we will view $\|\x\|_{0}$ as the number that counts how many nonzero entries $\x$ has.

To make writing in the rest of the paper easier, we will assume the
so-called \emph{linear} regime, i.e. we will assume that $k=\beta n$
and that the number of equations is $m=\alpha n$ where
$\alpha$ and $\beta$ are constants independent of $n$ (more
on the non-linear regime, i.e. on the regime when $m$ is larger than
linearly proportional to $k$ can be found in e.g.
\cite{CoMu05,GiStTrVe06,GiStTrVe07}). Of course, we do mention that all of our results can easily be adapted to various nonlinear regimes as well.

Looking back at (\ref{eq:system}), clearly one can consider an exhaustive search type of solution where one would look at all subsets of $k$ columns of $A$ and then attempt to solve the resulting system. However, in the linear regime that we assumed above such an approach becomes prohibitively slow as $n$ grows. That of course led in last several decades towards a search for more clever algorithms for solving (\ref{eq:system}). Many great algorithms were developed (especially during the last decade) and many of them have even provably excellent performance measures (see, e.g. \cite{JATGomp,JAT,NeVe07,DTDSomp,NT08,DaiMil08,DonMalMon09}).

A particularly successful technique for solving (\ref{eq:system}) that will be of our interest in this paper is a linear programming relaxation of (\ref{eq:l0}), called $\ell_1$-optimization. (Variations of the standard $\ell_1$-optimization from e.g.
\cite{CWBreweighted,SChretien08,SaZh08}) as well as those from \cite{SCY08,FL08,GN03,GN04,GN07,DG08,StojnicLqThrBnds10,StojnicLiftLqThrBnds13} related to $\ell_q$-optimization, $0<q<1$
are possible as well.) Basic $\ell_1$-optimization algorithm finds $\x$ in
(\ref{eq:system}) or (\ref{eq:l0}) by solving the following $\ell_1$-norm minimization problem
\begin{eqnarray}
\mbox{min} & & \|\x\|_{1}\nonumber \\
\mbox{subject to} & & A\x=\y. \label{eq:l1}
\end{eqnarray}
If one looks at $\|\x\|_{q},0\leq q\leq 1$, then clearly for two limiting values of $q$ one obtains either (\ref{eq:l0}) or (\ref{eq:l1}). On the other hand, one would be tempted to believe that as $q$ moves from $0$ and starts increasing towards $1$, $\|\x\|_{q}$ deviates more and more from the desired objective given in (\ref{eq:l0}) and approaches closer and closer towards the objective given in (\ref{eq:l1}). The reason why one is typically interested in value $q=1$ (and consecutively in taking $\|\x\|_{1}$ as the objective in (\ref{eq:l1})) is because in that case the resulting optimization problem given in (\ref{eq:l1}) is known to be solvable in polynomial time.

Due to its popularity the literature on the use of the above algorithm is rapidly growing. We below restrict our attention to two, in our mind, the most influential works that relate to (\ref{eq:l1}).

The first one is \cite{CRT} where the authors were able to show that if
$\alpha$ and $n$ are given, $A$ is given and satisfies the restricted isometry property (RIP) (more on this property the interested reader can find in e.g. \cite{Crip,CRT,Bar,Ver,ALPTJ09,BCTsharp09,BCTdecay09,BCTT09,BT10,StojnicRicBnds13}), then
any unknown vector $\tilde{\x}$ with no more than $k=\beta n$ (where $\beta$
is a constant dependent on $\alpha$ and explicitly
calculated in \cite{CRT}) non-zero elements can be recovered by
solving (\ref{eq:l1}). As earlier, this assumes that $\y$ in (\ref{eq:l1}) was in
fact generated by that $\tilde{\x}$ (in fact, to be more accurate, by the product $A\tilde{\x}$) and given to us.

However, the RIP is only a \emph{sufficient}
condition for $\ell_1$-optimization to produce the $k$-sparse solution of
(\ref{eq:system}). Instead of characterizing $A$ through the RIP
condition, in \cite{DonohoUnsigned,DonohoPol} Donoho looked at its geometric properties/potential. Namely,
in \cite{DonohoUnsigned,DonohoPol} Donoho considered the polytope obtained by
projecting the regular $n$-dimensional cross-polytope $C_p^n$ by $A$. He then established that
the solution of (\ref{eq:l1}) will be the $k$-sparse solution of
(\ref{eq:system}) if and only if
$AC_p^n$ is centrally $k$-neighborly
(for the definitions of neighborliness, details of Donoho's approach, and related results the interested reader can consult now already classic references \cite{DonohoUnsigned,DonohoPol,DonohoSigned,DT}). In a nutshell, using the results
of \cite{PMM,AS,BorockyHenk,Ruben,VS}, it is shown in
\cite{DonohoPol}, that if $A$ is a random $m\times n$
ortho-projector matrix then with overwhelming probability $AC_p^n$ is centrally $k$-neighborly (as usual, under overwhelming probability we in this paper assume
a probability that is no more than a number exponentially decaying in $n$ away from $1$). Miraculously, \cite{DonohoPol,DonohoUnsigned} provided a precise characterization of $m$ and $k$ (in a large dimensional context) for which this happens.

In a series of our own work (see, e.g. \cite{StojnicICASSP09,StojnicCSetam09,StojnicUpper10}) we then created an alternative probabilistic approach which was also capable of providing the precise characterization between $m$ and $k$ that guarantees success/failure of (\ref{eq:l1}) when used for finding the $k$-sparse solution of (\ref{eq:system}). The approach was a combination of geometric and purely probabilistic ideas. The following theorem summarizes the results we obtained in e.g. \cite{StojnicICASSP09,StojnicCSetam09,StojnicUpper10,StojnicICASSP10var}.

\begin{theorem}(Exact threshold)
Let $A$ be an $m\times n$ matrix in (\ref{eq:system})
with i.i.d. standard normal components. Let
the unknown $\x$ in (\ref{eq:system}) be $k$-sparse. Further, let the location and signs of nonzero elements of $\x$ be arbitrarily chosen but fixed.
Let $k,m,n$ be large
and let $\alpha=\frac{m}{n}$ and $\beta_w=\frac{k}{n}$ be constants
independent of $m$ and $n$. Let $\erfinv$ be the inverse of the standard error function associated with zero-mean unit variance Gaussian random variable.  Further,
let all $\epsilon$'s below be arbitrarily small constants.
\begin{enumerate}
\item Let $\htheta_w$, ($\beta_w\leq \htheta_w\leq 1$) be the solution of
\begin{equation}
(1-\epsilon_{1}^{(c)})(1-\beta_w)\frac{\sqrt{\frac{2}{\pi}}e^{-(\erfinv(\frac{1-\theta_w}{1-\beta_w}))^2}}{\theta_w}-\sqrt{2}\erfinv ((1+\epsilon_{1}^{(c)})\frac{1-\theta_w}{1-\beta_w})=0.\label{eq:thmweaktheta}
\end{equation}
If $\alpha$ and $\beta_w$ further satisfy
\begin{equation}
\alpha>\frac{1-\beta_w}{\sqrt{2\pi}}\left (\sqrt{2\pi}+2\frac{\sqrt{2(\erfinv(\frac{1-\htheta_w}{1-\beta_w}))^2}}{e^{(\erfinv(\frac{1-\htheta_w}{1-\beta_w}))^2}}-\sqrt{2\pi}
\frac{1-\htheta_w}{1-\beta_w}\right )+\beta_w
-\frac{\left ((1-\beta_w)\sqrt{\frac{2}{\pi}}e^{-(\erfinv(\frac{1-\hat{\theta}_w}{1-\beta_w}))^2}\right )^2}{\hat{\theta}_w}\label{eq:thmweakalpha}
\end{equation}
then with overwhelming probability the solution of (\ref{eq:l1}) is the $k$-sparse $\x$ from (\ref{eq:system}).
\item Let $\htheta_w$, ($\beta_w\leq \htheta_w\leq 1$) be the solution of
\begin{equation}
(1+\epsilon_{2}^{(c)})(1-\beta_w)\frac{\sqrt{\frac{2}{\pi}}e^{-(\erfinv(\frac{1-\theta_w}{1-\beta_w}))^2}}{\theta_w}-\sqrt{2}\erfinv ((1-\epsilon_{2}^{(c)})\frac{1-\theta_w}{1-\beta_w})=0.\label{eq:thmweaktheta1}
\end{equation}
If on the other hand $\alpha$ and $\beta_w$ satisfy
\begin{multline}
\hspace{-.5in}\alpha<\frac{1}{(1+\epsilon_{1}^{(m)})^2}\left ((1-\epsilon_{1}^{(g)})(\htheta_w+\frac{2(1-\beta_w)}{\sqrt{2\pi}} \frac{\sqrt{2(\erfinv(\frac{1-\htheta_w}{1-\beta_w}))^2}}{e^{(\erfinv(\frac{1-\htheta_w}{1-\beta_w}))^2}})
-\frac{\left ((1-\beta_w)\sqrt{\frac{2}{\pi}}e^{-(\erfinv(\frac{1-\hat{\theta}_w}{1-\beta_w}))^2}\right )^2}{\hat{\theta}_w(1+\epsilon_{3}^{(g)})^{-2}}\right )\label{eq:thmweakalpha}
\end{multline}
then with overwhelming probability there will be a $k$-sparse $\x$ (from a set of $\x$'s with fixed locations and signs of nonzero components) that satisfies (\ref{eq:system}) and is \textbf{not} the solution of (\ref{eq:l1}).
\end{enumerate}
\label{thm:thmweakthr}
\end{theorem}
\begin{proof}
The first part was established in \cite{StojnicCSetam09} and the second one was established in \cite{StojnicUpper10}. An alternative way of establishing the same set of results was also presented in \cite{StojnicEquiv10}.
\end{proof}

We below provide a more informal interpretation of what was established by the above theorem. Assume the setup of the above theorem. Let $\alpha_w$ and $\beta_w$ satisfy the following:

\noindent \underline{\underline{\textbf{Fundamental characterization of the $\ell_1$ performance:}}}

\begin{center}
\shadowbox{$
(1-\beta_w)\frac{\sqrt{\frac{2}{\pi}}e^{-(\erfinv(\frac{1-\alpha_w}{1-\beta_w}))^2}}{\alpha_w}-\sqrt{2}\erfinv (\frac{1-\alpha_w}{1-\beta_w})=0.
$}
-\vspace{-.5in}\begin{equation}
\label{eq:thmweaktheta2}
\end{equation}
\end{center}

Then:
\begin{enumerate}
\item If $\alpha>\alpha_w$ then with overwhelming probability the solution of (\ref{eq:l1}) is the $k$-sparse $\x$ from (\ref{eq:system}).
\item If $\alpha<\alpha_w$ then with overwhelming probability there will be a $k$-sparse $\x$ (from a set of $\x$'s with fixed locations and signs of nonzero components) that satisfies (\ref{eq:system}) and is \textbf{not} the solution of (\ref{eq:l1}).
    \end{enumerate}

The above theorem (as well as corresponding results obtained earlier in \cite{DonohoPol,DonohoUnsigned})) essentially settles typically behavior of $\ell_1$ optimization when used for solving (\ref{eq:system}) or (\ref{eq:l0}). In this paper we will look at a couple of upgrades of the standard $\ell_1$ optimization from (\ref{eq:l1}). We will provide a rigorous analytical confirmation that these upgrades indeed improve on the performance of $\ell_1$ when it comes to the values of the recoverable sparsity (i.e. $\beta_w$). However, such an improvement will come with a price to pay. Namely, the algorithmic upgrades that we will consider will assume a certain amount of ``feedback", or in other words a certain amount of pre-knowledge about the problem at hand. Consequently, there will be two natural takeaways: 1) when such a pre-knowledge is available the upgraded versions will be superior to the standard $\ell_1$ (a fact clearly expected) and 2) one then may be able to translate the original problem (\ref{eq:system}) to a somewhat different problem that accounts for ability of providing the needed ``feedback" (a fact probably expected but here precisely characterized as well).

We organize the rest of the paper in the following way. In Section
\ref{sec:knownsupp} we introduce the first of the two above mentioned upgraded versions of (\ref{eq:l1}) and provide its a performance analysis in a statistical context. In Section \ref{sec:hidknownsupp} we then present the second one together with its a performance analysis. Finally, in Sections \ref{sec:discussion} and \ref{sec:conc} we discuss obtained results and their potential value.

\section{Partially known support}
\label{sec:knownsupp}

In this section we will look at a slightly modified version of the problem from (\ref{eq:system}) (or (\ref{eq:l0})). We start by recalling that what makes the problem in (\ref{eq:l1}) hard is determining the location of nonzero components of $\x$ (from this point on, we will often refer to these locations as the support of vector $\x$ and occasionally may even denote it as $supp(\x)$). One then may wonder if there was a way to determine some of these locations would then be possible to recover a higher sparsity by using (\ref{eq:l1}) or its a slight modification. The analysis that we will present below will provide a positive answer to this question. Moreover, depending on how many of these locations are \emph{a priori} known one can actually precisely quantify what type of improvement over standard $\ell_1$ from (\ref{eq:l1}) can be expected. Before proceeding with the analysis we first introduce several mathematical terms that we will often use.

We start by introducing vectors with \emph{partially} known support, see Figure \ref{fig:model2} (more on this type of vectors as well as on their potential applications can be found in e.g. \cite{VasLu09}).
\begin{figure}[htb]
\centering
\centerline{\epsfig{figure=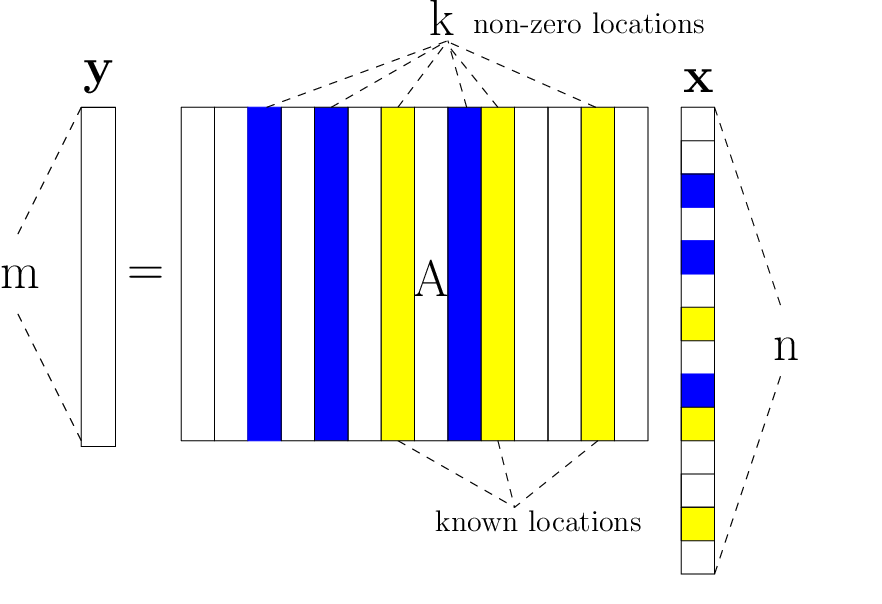,width=10.5cm,height=6cm}}
\caption{Model of a linear system; vector $\x$ is $k$-sparse; some of the non-zero locations are known}
\label{fig:model2}
\end{figure}
Under partially known support we assume that \emph{locations} of a fraction of the non-zero components of $\x$ are \emph{a priori} known and that that knowledge can possibly be exploited in the design of the recovery algorithms. To make everything precise, we will denote by $\Pi$ the set of the indexes of the known locations of the non-zero components of $\x$. We will further denote by  $\eta k$ (where $\eta$ is a constant independent of $n$ and $0\leq \eta\leq 1$) the cardinality of set $\Pi$. To recover $k$-sparse $\x$ with partially known support one can then design the algorithms that would attempt to exploit the available ``feedback", i.e. the available information about known locations of a fraction of nonzero components of $\x$. There are obviously many ways how one can attempt to do so . Here, we will focus on a particular modification of (\ref{eq:l1}) considered in e.g. \cite{VasLu09,SteChr09,StojnicICASSP10knownsupp}. Such a modification assumes the following slight adjustment to (\ref{eq:l1})
\begin{eqnarray}
\mbox{min} & & \sum_{i\notin\Pi} |\x_i|\nonumber \\
\mbox{subject to} & & A\x=\y. \label{eq:l1imp}
\end{eqnarray}
We will on occasion refer to the above adjustment of the $\ell_1$ as the \emph{partial} $\ell_1$. One then expects that the above algorithm will perform better than the standard $\ell_1$ from (\ref{eq:l1}). Below we will provide an analysis that will confirm such an expectation in a statistical context. Moreover, it will precisely quantify by how much the algorithm from (\ref{eq:l1imp}) improves on its a counterpart from (\ref{eq:l1}).

\subsection{Performance analysis of (\ref{eq:l1imp})}
\label{sec:analknownsupp}

In this section we will attempt to obtain the results qualitatively similar to those presented in Theorem \ref{thm:thmweakthr}. Of course, the results presented in Theorem \ref{thm:thmweakthr} are related to performance of (\ref{eq:l1}), whereas here we will try to create their an analogue that relates to (\ref{eq:l1imp}). As mentioned earlier, the results presented in Theorem \ref{thm:thmweakthr} were obtained in a series of work \cite{StojnicICASSP09,StojnicCSetam09,StojnicUpper10}. Below, we adapt some of these results so that they can handle the problems of interest here. In doing so, we will in this and all subsequent sections assume a substantial level of familiarity with many of the well-known results that relate to the performance characterization of (\ref{eq:l1}) (we will fairly often recall on many results/definitions that we established in \cite{StojnicICASSP09,StojnicCSetam09,StojnicUpper10,StojnicICASSP10knownsupp}).

Before proceeding further with a detail presentation we briefly sketch what specifically we will be interested in showing below. Namely, using the analysis of \cite{StojnicCSetam09,StojnicUpper10} mentioned earlier, for a specific group of randomly generated matrices $A$, one can determine values $\beta_w^{(p)}$ for the entire range of $\alpha$, i.e. for $0\leq \alpha\leq 1$, where $\beta_{w}^{(p)}$ is the
maximum allowable value of $\beta$ such that
(\ref{eq:l1imp}) finds the $k$-sparse solution of (\ref{eq:system}) with overwhelming
probability for \emph{any} $k$-sparse $\x$ with given fixed locations of non-zero components, a given fixed combination of its elements signs, and \emph{a priori} known location of a fraction $\eta$ of its non-zero components. (As discussed in \cite{StojnicICASSP10var,StojnicICASSP09,StojnicCSetam09,StojnicICASSP10knownsupp}, this value of $\beta_w^{(p)}$ is often referred to as the \emph{weak} threshold.) Clearly, $\beta_w^{(p)}$ will be a function of the fraction of known support $\eta$. Furthermore, (as expected) it will turn out that as $\eta$ increases the values of $\beta_w^{(p)}$ increase as well. This, in other words, means that a larger number of known non-zero locations implies a higher recoverable sparsity in (\ref{eq:l1imp}).

We are now ready to start the analysis. We begin by recalling on a theorem from \cite{StojnicICASSP10knownsupp} that provides a characterization as to when the solution of (\ref{eq:l1imp}) is $\tilde{\x}$, i.e. the $k$ sparse solution of (\ref{eq:system}) or (\ref{eq:l0}). Since the analysis will clearly be irrelevant with respect to what particular locations and what particular combination of signs of nonzero elements are chosen, we can for the simplicity of the exposition and without loss of generality assume that the components $\x_{1},\x_{2},\dots,\x_{n-k}$ of $\x$ are equal to zero and the components $\x_{n-k+1},\x_{n-k+2},\dots,\x_n$ of $\x$ are smaller than or equal to zero. Also, we will assume that it is \emph{a priori} known that $\x_{n-\eta k+1},\x_{n-\eta k+2},\dots,\x_n$ are among the $k$ non-zero components of $\x$ (one should note that while for our analysis it is assumed that $\x_{1},\x_{2},\dots,\x_{n-k}$ of $\x$ are equal to zero this fact is not known to the algorithm given in (\ref{eq:l1imp})). This essentially means that in(\ref{eq:l1imp}) one has $\Pi=\{n-\eta k+1,n-\eta k+2,\dots,n\}$.
Under these assumptions we have the following lemma (similar characterizations adopted in different contexts can be found in
\cite{DH01,XHapp,SPH,StojnicICASSP09,GN03}).
\begin{lemma}(Nonzero part of $\x$ has fixed signs, location; location of fraction $\eta$ of non-zero part is known)
Assume that an $m\times n$ measurement matrix $A$ is given. Let $\x$
be a $k$-sparse vector whose nonzero components are negative and let $\x_1=\x_2=\dots=\x_{n-k}=0$. Also, let it be known to the algorithm given in (\ref{eq:l1}) that $\x_{n-\eta k+1},\x_{n-\eta k+2},\dots,\x_n$ are among the $k$ non-zero components of $\x$, i.e. let $\Pi=\{n-\eta k+1,n-\eta k+2,\dots,n\}$, where $0\leq\eta\leq 1$.
Further, assume that $\y=A\x$ and that $\w$ is
an $n\times 1$ vector. If
\begin{equation}
(\forall \w\in \textbf{R}^n | A\w=0) \quad  \sum_{i=n-k+1}^{n-\eta k} \w_i<\sum_{i=1}^{n-k}|\w_{i}|,
\end{equation}
then the solutions of (\ref{eq:l1imp}) and (\ref{eq:l0}) (or (\ref{eq:system})) coincide. Moreover, if
\begin{equation}
(\exists \w\in \textbf{R}^n | A\w=0) \quad  \sum_{i=n-k+1}^{n-\eta k} \w_i\geq \sum_{i=1}^{n-k}|\w_{i}|,
\label{eq:thmeqgen}
\end{equation}
then there will be a $k$-sparse nonpositive $\x$ that satisfies (\ref{eq:system}) and is not the solution of (\ref{eq:l1imp}).
\label{lemma:lemmaknownsuppcond}
\end{lemma}
\begin{proof}
Follows directly from the corresponding results in \cite{StojnicICASSP09,StojnicCSetam09,StojnicUpper10}.
\end{proof}
Having matrix $A$ such that
(\ref{eq:thmeqgen}) holds would be enough for solutions of (\ref{eq:l1}) and (\ref{eq:l0}) (or (\ref{eq:system})) to coincide. If one
assumes that $m$ and $k$ are proportional to $n$ (the case of our
interest in this paper) then the construction of the deterministic
matrices $A$ that would satisfy
(\ref{eq:thmeqgen}) is not an easy task (in fact, one may say that together with the ones that correspond to the standard $\ell_1$ it is one of the most fundamental open problems in the area of theoretical compressed sensing). However, turning to
random matrices significantly simplifies things. That is the route that will pursuit below. In fact to be a bit more specific, we will assume that the elements of matrix $A$ are i.i.d. standard normal random variables. All results that we will present below will hold for many other types of randomness (we will discuss this in more detail in Section \ref{sec:conc}). However, to make the presentation as smooth as possible we assume the standard Gaussian scenario.

We then follow the strategy of \cite{StojnicCSetam09}. To that end we will make use of the following theorem:
\begin{theorem}(\cite{Gordon88} Escape through a mesh)
\label{thm:Gordonmesh} Let $S$ be a subset of the unit Euclidean
sphere $S^{n-1}$ in $R^{n}$. Let $Y$ be a random
$(n-m)$-dimensional subspace of $R^{n}$, distributed uniformly in
the Grassmanian with respect to the Haar measure. Let
\begin{equation}
w(S)=E\sup_{\w\in S} (\h^T\w) \label{eq:widthdef}
\end{equation}
where $\h$ is a random column vector in $R^{n}$ with i.i.d. ${\cal
N}(0,1)$ components. Assume that
$w(S)<\left ( \sqrt{m}-\frac{1}{4\sqrt{m}}\right )$. Then
\begin{equation}
P(Y\cap S= \emptyset )>1-3.5e^{-\frac{\left (\sqrt{m}-\frac{1}{4\sqrt{m}}-w(S) \right ) ^2}{18}}.
\label{eq:thmesh}
\end{equation}
\end{theorem}

As mentioned above, to make use of Theorem \ref{thm:Gordonmesh} we follow the strategy presented in \cite{StojnicCSetam09}. We start by defining a set $\Swp$
\begin{equation}
\Swp=\{\w\in S^{n-1}| \quad \sum_{i=n-k+1}^{n-\eta k} \w_i<\sum_{i=1}^{n-k}|\w_{i}|\},\label{eq:defSwp}
\end{equation}
where $S^{n-1}$ is the unit sphere in $R^n$. The strategy of \cite{StojnicCSetam09} assumes roughly the following: if $w(\Swp)< \sqrt{m}-\frac{1}{4\sqrt{m}}$ is positive with overwhelming probability for certain combination of $k$, $m$, and $n$ then for $\alpha=\frac{m}{n}$ one has a lower bound $\beta_{w}^{(p)}=\frac{k}{n}$ on the true value of the \emph{weak} threshold with overwhelming probability (under overwhelming probability we of course assume a probability that is no more than a number exponentially decaying in $n$ away from $1$). More on the definition of the weak threshold the interested reader can find in e.g. \cite{StojnicCSetam09,StojnicICASSP09,DonohoPol}. The above basically means that if one can handle $w(\Swp)$ then, when $n$ is large one can, roughly speaking, use the condition $w(\Swp)< \sqrt{m}$ to obtain an attainable lower bound $\beta_{w}^{(p)}$ for any given $0<\alpha\leq 1$.

To that end we then look at
\begin{equation}
w(S_{w}^{(p)})=E\max_{\w\in S_{w}^{(p)}} (\h^T\w), \label{eq:widthdefswp}
\end{equation}
where we have replaced the $\sup$ from (\ref{eq:widthdef}) with a $\max$. Following further what was done in \cite{StojnicCSetam09,StojnicICASSP10knownsupp} one then can write
\begin{eqnarray}
w(\Swp) & = & E\max_{\w\in \Swp} (\h^T\w)=E\max_{\w\in \Swp} (\sum_{i=1}^{n-k} |\h_i\w_i|+\sum_{i=n-k+1}^{n}\h_i\w_i) \nonumber \\
& = & E\max_{\w\in \Swp} (\sum_{i=1}^{n-k} |\h_i||\w_i|+\sum_{i=n-k+1}^{n}\h_i\w_i).\label{eq:workww0}
\end{eqnarray}
Let $\h_{1:(n-k)}=(\h_1,\h_2,\dots,\h_{n-k})^T$. Further, let now
$|\h|_{(i)}^{(n-k)}$ be the $i$-th smallest magnitude of elements of $\h_{1:(n-k)}$. Set
$\hw=(|\h|_{(1)}^{(n-k)},|\h|_{(2)}^{(n-k)},\dots,|\h|_{(n-k)}^{(n-k)},\h_{n-k+1},\h_{n-k+2},\dots,\h_{n})^T$.\\
Then one can simplify (\ref{eq:workww0}) in the following way
\begin{eqnarray}
w(\Swp) = E\max_{\t\in R^n} & &  \hw^T\t \nonumber \\
\mbox{subject to} &  & \t_i\geq 0, 0\leq i\leq (n-k)\nonumber \\
& & \sum_{i=n-k+1}^{n-\eta k}\t_i\geq \sum_{i=1}^{n-k} \t_i \nonumber \\
& & \sum_{i=1}^n\t_i^2\leq 1.\label{eq:workww1}
\end{eqnarray}
Let $\z\in R^n$ be a column vector such that $\z_i=1,1\leq i\leq (n-k)$, $\z_i=-1,n-k+1\leq i\leq n-\eta k$, and $\z_i=0,n-\eta k+1\leq i\leq n$. Following step by step the derivation in \cite{StojnicCSetam09} one has, based on the Lagrange duality theory, that there is a $\cweak=(1-\theta_w)n\leq (n-k)$ such that
\begin{multline}
\lim_{n\rightarrow \infty} \frac{w(\Swp)}{\sqrt{n}}=\lim_{n\rightarrow \infty} \frac{E\max_{\w\in S_{w}^{(p)}} (\h^T\w)}{\sqrt{n}}\approxeq
\sqrt{\lim_{n\rightarrow \infty} \frac{E\sum_{i=\cweak+1}^n\hw_i^2}{n}-\frac{(\lim_{n\rightarrow \infty} \frac{E(\hw^T\z)-E\sum_{i=1}^{\cweak}\hw_i}{n})^2}{1-\lim_{n\rightarrow \infty}\frac{\cweak}{n}-\eta \lim_{n\rightarrow \infty}\frac{k}{n}}}.\\\label{eq:wwp}
\end{multline}
where $\hw_i$ is the $i$-th element of vector $\hw$. Moreover, \cite{StojnicCSetam09} also establishes the way to determine a critical $\cweak$. Roughly speaking it establishes the following identity
\begin{equation}
\frac{(\lim_{n\rightarrow \infty} \frac{E(\hw^T\z)-E\sum_{i=1}^{\cweak}\hw_i}{n})^2}{1-\lim_{n\rightarrow \infty}\frac{\cweak}{n}-\eta \lim_{n\rightarrow \infty}\frac{k}{n}}\approxeq \lim_{n\rightarrow \infty} \frac{E\hw_{\cweak}}{n}.\label{eq:condcwp}
\end{equation}
Using further the technique of \cite{StojnicCSetam09} one can actually explicitly characterize $w(\Swp)$ in (\ref{eq:wwp}) in the following way:
\begin{multline}
\hspace{-.5in}\left ( \lim_{n\rightarrow \infty} \frac{w(\Swp)}{\sqrt{n}} \right )^{2}=\left ( \lim_{n\rightarrow \infty} \frac{E\max_{\w\in S_{w}^{(p)}} (\h^T\w)}{\sqrt{n}} \right )^{2}\approxeq
 \frac{1-\beta_w^{(p)}}{\sqrt{2\pi}}\left (\sqrt{2\pi}+2\frac{\sqrt{2(\erfinv(\frac{1-\htheta_w}{1-\beta_w^{(p)}}))^2}}{e^{(\erfinv(\frac{1-\htheta_w}{1-\beta_w^{(p)}}))^2}}-\sqrt{2\pi}
\frac{1-\htheta_w}{1-\beta_w^{(p)}}\right )\\+\beta_w^{(p)}
-\frac{\left ((1-\beta_w^{(p)})\sqrt{\frac{2}{\pi}}e^{-(\erfinv(\frac{1-\hat{\theta}_w}{1-\beta_w^{(p)}}))^2}\right )^2}{\hat{\theta}_w-\eta\beta_w^{(p)}},\label{eq:wwp1}
\end{multline}
where $\htheta_w$ is the solution of
\begin{equation}
\frac{\left ((1-\beta_w^{(p)})\sqrt{\frac{2}{\pi}}e^{-(\erfinv(\frac{1-\theta_w}{1-\beta_w^{(p)}}))^2}\right )}{\theta_w-\eta\beta_w^{(p)}}\approxeq
\sqrt{2}\mbox{erfinv}\left (\frac{1-\theta_w}{1-\beta_w^{(p)}}\right ).\label{eq:condcwp1}
\end{equation}

We summarize the above results in the following theorem.

\begin{theorem}(Location of fraction $\eta$ of non-zero elements is known)
Let $A$ be an $m\times n$ measurement matrix in (\ref{eq:system})
with the null-space uniformly distributed in the Grassmanian. Let
the unknown $\x$ in (\ref{eq:system}) be $k$-sparse. Further, let the location and signs of nonzero elements of $\x$ be arbitrarily chosen but fixed. Assume that the location of $\eta k$ of non-zero elements is a priori known and let $\Pi$ be the set of those locations.
Let $k,m,n$ be large
and let $\alpha=\frac{m}{n}$, $\beta_w^{(p)}=\frac{k}{n}$, and $\eta$ be constants
independent of $m$, $n$, and $k$. Let $\erfinv$ be the inverse of the standard error function associated with zero-mean unit variance Gaussian random variable.  Further,
let $\htheta_w$, ($\beta_w^{(p)}\leq \htheta_w\leq 1$) be the solution of
\begin{equation}
(1-\beta_w^{(p)})\frac{\sqrt{\frac{2}{\pi}}e^{-(\erfinv(\frac{1-\theta_w}{1-\beta_w^{(p)}}))^2}}{\theta_w-\eta\beta_w^{(p)}}\approxeq\sqrt{2}\erfinv \frac{1-\theta_w}{1-\beta_w^{(p)}}).\label{eq:thmweaktheta}
\end{equation}

1) If $\alpha$ and $\beta_w^{(p)}$ further satisfy
\begin{equation}
\alpha> \frac{1-\beta_w^{(p)}}{\sqrt{2\pi}}\left (\sqrt{2\pi}+2\frac{\sqrt{2(\erfinv(\frac{1-\htheta_w}{1-\beta_w^{(p)}}))^2}}{e^{(\erfinv(\frac{1-\htheta_w}{1-\beta_w^{(p)}}))^2}}-\sqrt{2\pi}
\frac{1-\htheta_w}{1-\beta_w^{(p)}}\right )+\beta_w^{(p)}
-\frac{\left ((1-\beta_w^{(p)})\sqrt{\frac{2}{\pi}}e^{-(\erfinv(\frac{1-\hat{\theta}_w}{1-\beta_w^{(p)}}))^2}\right )^2}{\hat{\theta}_w-\eta\beta_w^{(p)}},\label{eq:thmweakalpha1}
\end{equation}
then the solution of (\ref{eq:l1imp}) and the solution of (\ref{eq:l0}) (or the $k$-sparse solution $\tilde{\x}$ of (\ref{eq:system})) coincide with overwhelming
probability.

2) If $\alpha$ and $\beta_w^{(p)}$ are such that
\begin{equation}
\alpha<\frac{1-\beta_w^{(p)}}{\sqrt{2\pi}}\left (\sqrt{2\pi}+2\frac{\sqrt{2(\erfinv(\frac{1-\htheta_w}{1-\beta_w^{(p)}}))^2}}{e^{(\erfinv(\frac{1-\htheta_w}{1-\beta_w^{(p)}}))^2}}-\sqrt{2\pi}
\frac{1-\htheta_w}{1-\beta_w^{(p)}}\right )+\beta_w^{(p)}
-\frac{\left ((1-\beta_w^{(p)})\sqrt{\frac{2}{\pi}}e^{-(\erfinv(\frac{1-\hat{\theta}_w}{1-\beta_w^{(p)}}))^2}\right )^2}{\hat{\theta}_w-\eta\beta_w^{(p)}},\label{eq:thmweakalpha2}
\end{equation}
then with overwhelming probability there will be a $k$-sparse $\x$ (from a set of $\x$'s with fixed locations and signs of nonzero components) that satisfies (\ref{eq:system}) and is \textbf{not} the solution of (\ref{eq:l1}).
\label{thm:thmweakthrknownsupp}
\end{theorem}
\begin{proof}
The first part follows from the above discussion as well as from the considerations presented in \cite{StojnicICASSP10knownsupp} and the analysis presented in \cite{StojnicCSetam09}. The moreover part follows by a combination of the moreover part of Lemma \ref{lemma:lemmaknownsuppcond} and the considerations presented in \cite{StojnicUpper10,StojnicGorEx10}.
\end{proof}

\noindent \textbf{Remark:} To make writing easier in the previous theorem we removed all $\epsilon$'s used in Theorem \ref{thm:thmweakthr}.

In a more informal language one has the following interpretation of the above theorem. Assume the setup of the above theorem. Let $\alpha_w^{(p)}$ and $\beta_w^{(p)}$ satisfy the following:

\noindent \underline{\underline{\textbf{Fundamental characterization of the \emph{partial} $\ell_1$ performance:}}}

\begin{center}
\shadowbox{$
(1-\beta_w^{(p)})\frac{\sqrt{\frac{2}{\pi}}e^{-(\erfinv(\frac{1-\alpha_w^{(p)}}{1-\beta_w^{(p)}}))^2}}{\alpha_w^{(p)}-\eta\beta_w^{(p)}}-\sqrt{2}\erfinv (\frac{1-\alpha_w^{(p)}}{1-\beta_w^{(p)}})=0.
$}
-\vspace{-.5in}\begin{equation}
\label{eq:thmweakinfknownsupp}
\end{equation}
\end{center}

Then:
\begin{enumerate}
\item If $\alpha>\alpha_w^{(p)}$ then with overwhelming probability the solution of (\ref{eq:l1}) is the $k$-sparse $\x$ from (\ref{eq:system}).
\item If $\alpha<\alpha_w^{(p)}$ then with overwhelming probability there will be a $k$-sparse $\x$ (from a set of $\x$'s with fixed locations and signs of nonzero components) that satisfies (\ref{eq:system}) and is \textbf{not} the solution of (\ref{eq:l1}).
    \end{enumerate}

The above theorem essentially settles typical behavior of the \emph{partial} $\ell_1$ optimization from (\ref{eq:l1}) when used for solving (\ref{eq:system}) or (\ref{eq:l0}) assuming that a fraction $\eta$ of nonzero locations of $\x$ is a priori known.

The results for the weak threshold obtained from the above theorem
are presented in Figure \ref{fig:weakknownsupp}. Case $\eta=0$ corresponds to the standard compressed sensing setup where no information about the location of the non-zero components of $\x$ is a priori available. The threshold values obtained in that case correspond to the ones computed in \cite{StojnicICASSP10var,StojnicCSetam09,StojnicUpper10,StojnicEquiv10} (and presented in Theorem \ref{thm:thmweakthr}) and of course to those computed in \cite{DonohoPol}. As $\eta$ increases more knowledge about $\x$ is available and one expects that the threshold values of the recoverable sparsity should be higher. As results presented in Figure \ref{fig:weakknownsupp} indicate, the values of the threshold recoverable by the modified partial $\ell_1$ optimization from (\ref{eq:l1imp}) are indeed higher as $\eta$ increases. Also, on the right side of Figure \ref{fig:weakknownsupp} we show experimental results that we discuss below.
 \begin{figure}[htb]
\begin{minipage}[b]{0.5\linewidth}
\centering
\centerline{\epsfig{figure=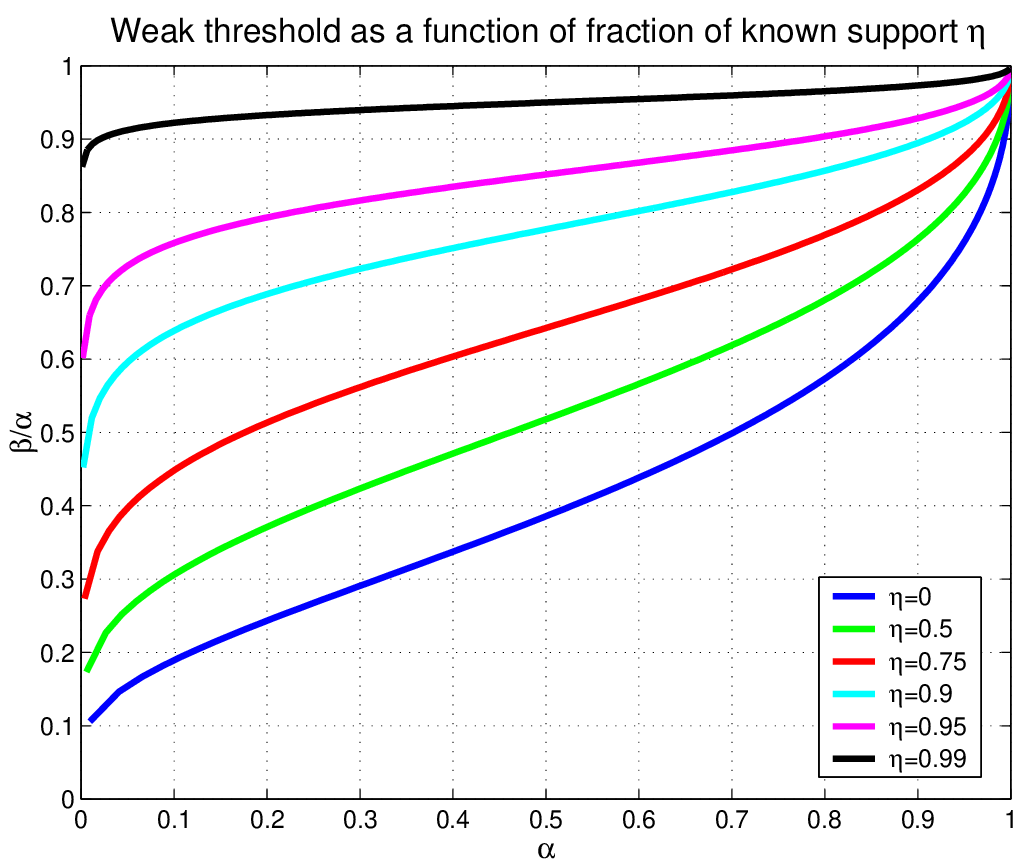,width=8cm,height=6cm}}
\end{minipage}
\begin{minipage}[b]{0.5\linewidth}
\centering
\centerline{\epsfig{figure=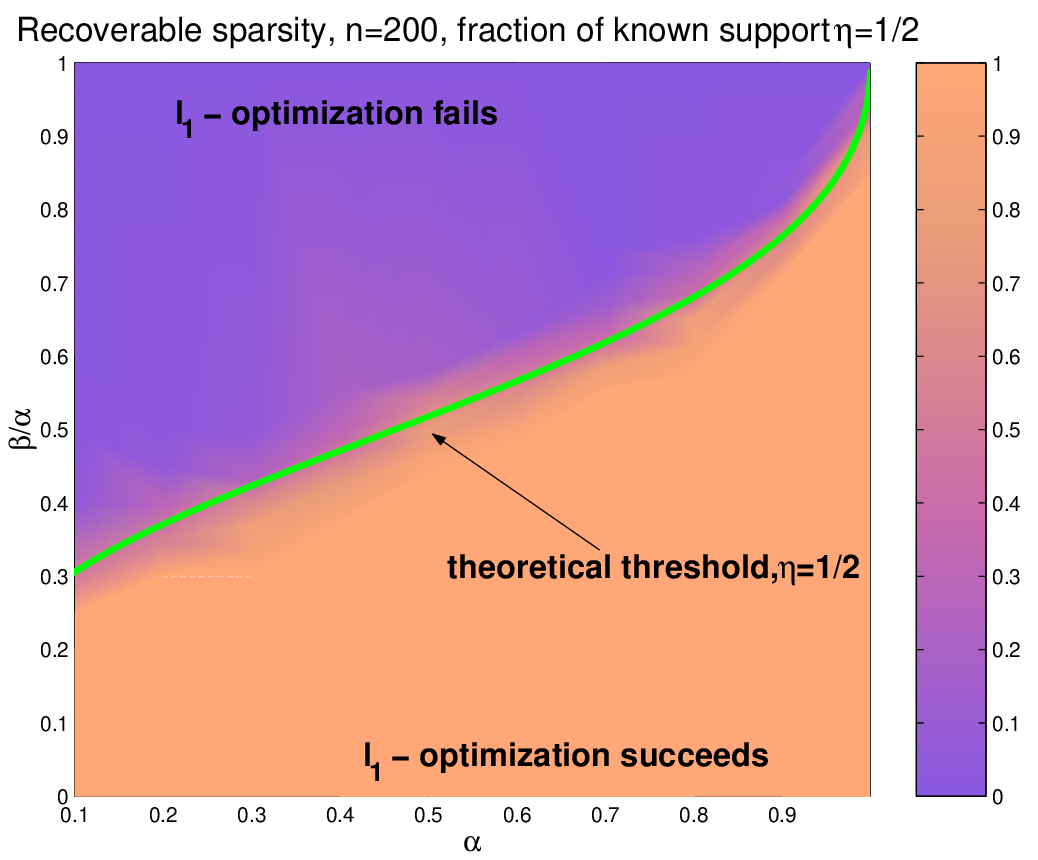,width=8cm,height=6cm}}
\end{minipage}
\caption{Left: Theoretical \emph{weak} threshold as a function of fraction of known support; Right: Experimentally recoverable sparsity; fraction of known support $\eta=\frac{1}{2}$}
\label{fig:weakknownsupp}
\end{figure}

\subsection{Numerical experiments}
\label{sec:simulp}

In this section we briefly discuss the results that we obtained from numerical experiments. In all our numerical experiments we fixed $n=200$ and $\eta=0.5$. We then generated matrices $A$ of size $m\times n$ with $m=(10,20,30,\dots,90,99)$. The components of the measurement matrices $A$ were generated as i.i.d. zero-mean unit variance Gaussian random variables.
For each $m$ we generated $k$-sparse signals $\x$ for several different values of $k$ from the transition zone (the locations of non-zero elements of $\x$ were chosen randomly and half or $\eta$ of them that correspond to the known part, i.e. to $\Pi$ in (\ref{eq:l1imp}), were chosen randomly as well). For each combination $(k,m)$ we generated $100$ different problem instances and recorded the number of times the partial $\ell_1$-optimization algorithm from (\ref{eq:l1imp}) failed to recover the correct $k$-sparse $\x$.
The obtained data are then interpolated and graphically presented on the right hand side of Figure \ref{fig:weakknownsupp}. The color of any point shows the probability of having partial $\ell_1$-optimization from (\ref{eq:l1imp}) succeed for a combination $(\alpha,\beta)$ that corresponds to that point. The colors are mapped to probabilities according to the scale on the right hand side of the figure.
The simulated results can naturally be compared to the theoretical prediction from Theorem \ref{thm:thmweakthrknownsupp}. Hence, we also show on the right hand side the theoretical value for the threshold calculated according to Theorem \ref{thm:thmweakthrknownsupp} (and obviously shown on the left hand side of the figure as well). We observe that the simulation results are in a good agreement with the theoretical calculation.

\section{\emph{Hidden} partially known support}
\label{sec:hidknownsupp}

In this section we will look at another slightly modified version of the problem from (\ref{eq:system}) (or (\ref{eq:l0})). As mentioned in the previous section, what makes the problem in (\ref{eq:l1}) hard is determining the location of nonzero components of $\x$. In the previous section we then looked at a bit relaxed scenario which in nutshell assumes the following: if there is a way to determine some of unknown locations then one should be able to recover a higher sparsity by using (\ref{eq:l1}) or its a slight modification. The analysis presented in the previous section then confirmed that if one uses for example (\ref{eq:l1imp}) instead of (\ref{eq:l1}) a higher sparsity is indeed recoverable. Moreover, depending on how many of these locations are \emph{a priori} known the analysis of the previous section precisely quantifies what type of improvement over standard $\ell_1$ from (\ref{eq:l1}), (\ref{eq:l1imp}) is expected to achieve.

Such a collection of results is then encouraging from the following point of view. Namely, if one can design an algorithm that provably recovers only a fraction of $supp(\x)$ then one can also guarantee an improvement over the standard $\ell_1$ from (\ref{eq:l1}). This in turn effectively translates the original sparse recovery problem from (\ref{eq:l0}) to its a possibly simpler version that only asks for a partial recovery. While such an understanding is conceptually correct, it contains a tiny problem. One has to be careful that for (\ref{eq:l1imp}) to be as successful as the results of the previous section predict, one should provide a set $\Pi$ that is known to contain only a subset of $supp(\x)$ (and basically nothing more than that). While designing algorithms that can provide a subset of $supp(\x)$ is not that hard (essentially any iterative upgrade of the standard $\ell_1$ from (\ref{eq:l1}) works in that way), it is substantially harder to insure that at the same time they do not provide anything else. In other words, it is more natural to expect that one can design algorithms that can provide a set of locations as an estimate of $supp(\x)$ such that it indeed does contain a fraction of elements in $supp(\x)$ but at the same time it also contains elements that are not in $supp(\x)$. This essentially means that typically all these iterative (or even not necessarily iterative) algorithms return a fraction of support of $\x$ hidden within a larger set of locations. For example, an algorithm can return a set of $k$ locations $\kappa$ that is an estimate for the support
of $\x$. Even when $\kappa$ does not match exactly $supp(\x)$ it may (and for almost any algorithm it will) still
contain some of the elements of $supp(\x)$. The difficulty is that one (differently from the previous section)
now does not know which of the locations are part of the support and which are not. If one knew which ones
are then obviously (\ref{eq:l1imp}) could be used for the recovery. However since this is not known one can not use (\ref{eq:l1imp})
directly. Perhaps surprisingly one can still benefit from having some of the support elements embedded in
the estimate $\kappa$. We in this section provide a precise characterization of such a benefit. However, before proceeding with the presentation we will first introduce a few mathematical definitions that we will need below.

We first introduce concept of vectors with \emph{hidden} partially known support (see Figure \ref{fig:model2hidden}). As usual let $\x$ be a $k$-sparse $n$-dimensional vector. Let $\kappa\subset \{1,2,\dots,n\}$ and let the cardinality of $\kappa$ be $k$ (we will for the simplicity choose $k$; however our results easily extend to any other value). Let $\Pi$ be the intersection of the set of nonzero locations of $\x$ ($supp(\x)$) and $\kappa$. As in the previous section, $\Pi$ is the set that is known to contain locations of some of the nonzero elements of $\x$. Differently though from what was the case in the previous section, $\Pi$ is not known now. What is known is $\kappa$ and the fact that $\Pi\in\kappa$. To make everything even more precise we will say that, as in the previous section, the cardinality of $\Pi$ is $\eta k$ (where $\eta$ is again a constant independent of $n$ and $0\leq \eta\leq 1$) and that $\x$ is a vector with \emph{hidden} partially known support. Moreover we will call $\kappa$ the estimate of $\x$'s support ($supp(\x)$).
\begin{figure}[htb]
\centering
\centerline{\epsfig{figure=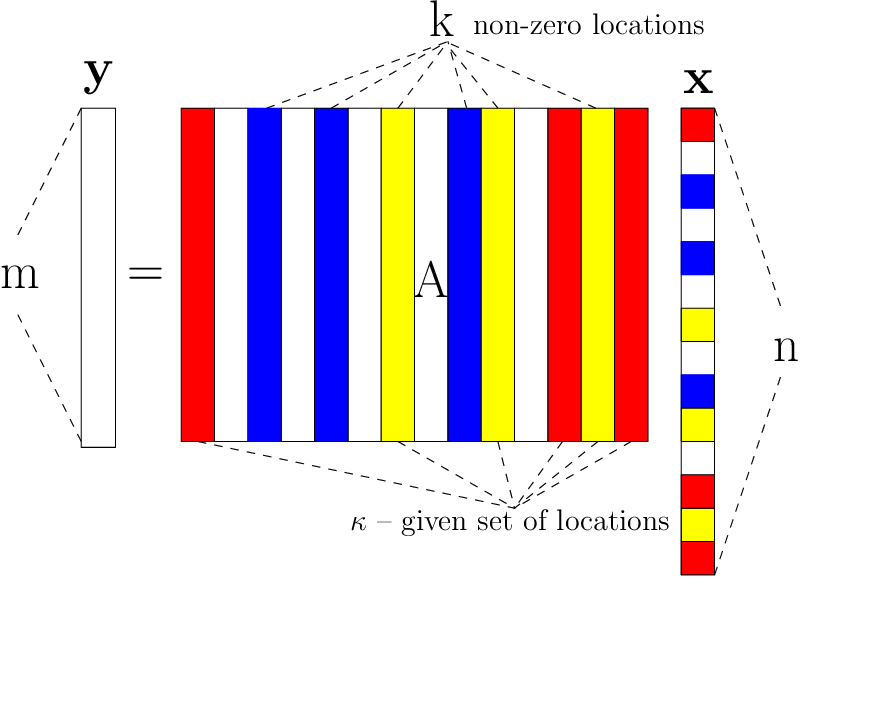,width=10.5cm,height=7cm}}
\vspace{-.4in}
\caption{Model of a linear system; vector $\x$ is $k$-sparse; some of the non-zero locations are in $\kappa$}
\label{fig:model2hidden}
\end{figure}
To recover $k$-sparse $\x$ with \emph{hidden} partially known support one can then design algorithms that would attempt to exploit the available ``feedback", i.e. the available information about known hidden locations of a fraction of nonzero components of $\x$. There are obviously many ways how one can attempt to do so. Here, we will focus on a particular modification of (\ref{eq:l1imp}) considered in the previous section (and ultimately in e.g. \cite{VasLu09,SteChr09,StojnicICASSP10knownsupp}). Such a modification assumes the following slight adjustment to (\ref{eq:l1imp})
\begin{eqnarray}
\mbox{min} & & \sum_{i\notin\kappa} |\x_i|\nonumber \\
\mbox{subject to} & & A\x=\y. \label{eq:l1imphidden}
\end{eqnarray}
We will on occasion refer to the above adjustment of the partial $\ell_1$ from (\ref{eq:l1imp}) as the \emph{hidden} partial $\ell_1$. One then expects that the above algorithm will perform better than the standard $\ell_1$ from (\ref{eq:l1}). Below we will provide an analysis that will confirm such an expectation in a statistical context. Moreover, it will precisely quantify by how much the algorithm from (\ref{eq:l1imphidden}) improves on its a counterpart from (\ref{eq:l1}).

\subsection{Performance analysis of (\ref{eq:l1imphidden})}
\label{sec:analknownsupphidden}

In this section we will attempt to obtain the results qualitatively similar to those presented in Theorems \ref{thm:thmweakthr} and \ref{thm:thmweakthrknownsupp}. Of course, the results presented in Theorems \ref{thm:thmweakthr} and \ref{thm:thmweakthrknownsupp} are related to performances of (\ref{eq:l1}) and (\ref{eq:l1imp}) respectively, whereas here we will try to create their an analogue that relates to (\ref{eq:l1imphidden}). As mentioned earlier, the results presented in Theorem \ref{thm:thmweakthr} were obtained in a series of work \cite{StojnicICASSP09,StojnicCSetam09,StojnicUpper10}. Below, we adapt some of these results as well as some of the results used in obtaining Theorem \ref{thm:thmweakthrknownsupp} so that they can handle the problems of interest here. In doing so, we will, as in the previous section, assume a substantial level of familiarity with many of the well-known results that relate to the performance characterization of (\ref{eq:l1}) (we will again fairly often recall on many results/definitions that we established in \cite{StojnicICASSP09,StojnicCSetam09,StojnicUpper10,StojnicICASSP10knownsupp}).

Before proceeding further with a detailed presentation we briefly recall on what specifically we will be interested in showing below. Namely, using the analysis of \cite{StojnicCSetam09,StojnicUpper10} mentioned earlier as well as what we presented in the previous section, for a specific group of randomly generated matrices $A$, one can determine values $\beta_w^{(h)}$ for the entire range of $\alpha$, i.e. for $0\leq \alpha\leq 1$, where $\beta_{w}^{(h)}$ is the
maximum allowable value of $\beta$ such that
(\ref{eq:l1imphidden}) finds the $k$-sparse solution of (\ref{eq:system}) with overwhelming
probability for \emph{any} $k$-sparse $\x$ with given fixed locations of non-zero components, a given fixed combination of its elements signs, and \emph{a priori} known to have a fraction $\eta$ of its non-zero components contained in a set of cardinality $k$, $\kappa$. As discussed in the previous section (and earlier of course in \cite{StojnicICASSP10var,StojnicICASSP09,StojnicCSetam09,StojnicICASSP10knownsupp}), this value of $\beta_w^{(h)}$ is often referred to as the \emph{weak} threshold. Clearly, $\beta_w^{(h)}$ will be a function of fraction $\eta$. Furthermore, (as expected) it will turn out that as $\eta$ increases the values of $\beta_w^{(h)}$ increase as well. This, in other words, means that a larger number of known (but hidden) non-zero locations implies a higher recoverable sparsity in (\ref{eq:l1imphidden}).

We are now ready to start the analysis. We again begin by establishing a lemma similar to Lemma \ref{lemma:lemmaknownsuppcond} (and of course to a corresponding theorem from \cite{StojnicICASSP10knownsupp}) that provides a characterization as to when the solution of (\ref{eq:l1imphidden}) is $\tilde{\x}$, i.e. the $k$ sparse solution of (\ref{eq:system}) or (\ref{eq:l0}). Since the analysis will again clearly be irrelevant with respect to what particular locations and what particular combination of signs of nonzero elements are chosen, we can for the simplicity of the exposition and without loss of generality again assume that the components $\x_{1},\x_{2},\dots,\x_{n-k}$ of $\x$ are equal to zero and the components $\x_{n-k+1},\x_{n-k+2},\dots,\x_n$ of $\x$ are smaller than or equal to zero. Also, we will assume that it is \emph{a priori} known that $\x_{n-\eta k+1},\x_{n-\eta k+2},\dots,\x_n$ are among the $k$ non-zero components of $\x$ (one should note that while for our analysis it is assumed that $\x_{1},\x_{2},\dots,\x_{n-k}$ of $\x$ are equal to zero this fact is not known to the algorithm given in (\ref{eq:l1imphidden})). This essentially means that one has $\Pi=\{n-\eta k+1,n-\eta k+2,\dots,n\}$. Moreover, without a loss of generality we will assume that $\kappa$ in (\ref{eq:l1imphidden}) is $\kappa=\{n-k-(1-\eta)k+1,n-k-(1-\eta)k+2,\dots,n-k,\Pi\}$ or in other words
$\kappa=\{n-k-(1-\eta)k+1,n-k-(1-\eta)k+2,\dots,n-k,n-\eta k+1,n-\eta k+2,\dots,n\}$.

Under these assumptions we have the following counterpart to Lemma \ref{lemma:lemmaknownsuppcond} (similar characterizations adopted in different contexts can be found in
\cite{DH01,XHapp,SPH,StojnicICASSP09,GN03}).
\begin{lemma}(\cite{StojnicICASSP10knownsupp} Nonzero part of $\x$ has fixed signs, location; location of fraction $\eta$ of non-zero part is known to be hidden in $\kappa$)
Assume that an $m\times n$ measurement matrix $A$ is given. Let $\x$
be a $k$-sparse vector whose nonzero components are negative and let $\x_1=\x_2=\dots=\x_{n-k}=0$. Also, let it be known to the algorithm given in (\ref{eq:l1}) that $\x_{n-\eta k+1},\x_{n-\eta k+2},\dots,\x_n$ are among the $k$ non-zero components of $\x$, i.e. let $\kappa=\{n-k-(1-\eta)k+1,n-k-(1-\eta)k+2,\dots,n-k,n-\eta k+1,n-\eta k+2,\dots,n\}$, where $0\leq\eta\leq 1$.
Further, assume that $\y=A\x$ and that $\w$ is
an $n\times 1$ vector. If
\begin{equation}
(\forall \w\in \textbf{R}^n | A\w=0) \quad  \sum_{i=n-k+1}^{n-\eta k} \w_i<\sum_{i=1}^{n-k-(1-\eta)k}|\w_{i}|,
\end{equation}
then the solutions of (\ref{eq:l1imphidden}) and (\ref{eq:l0}) (or (\ref{eq:system})) coincide. Moreover, if
\begin{equation}
(\exists \w\in \textbf{R}^n | A\w=0) \quad  \sum_{i=n-k+1}^{n-\eta k} \w_i\geq \sum_{i=1}^{n-k-(1-\eta)k}|\w_{i}|,
\label{eq:thmeqgenhidden}
\end{equation}
then there will be a $k$-sparse nonpositive $\x$ that satisfies (\ref{eq:system}) and is not the solution of (\ref{eq:l1imphidden}).
\label{lemma:lemmaknownsuppcondhidden}
\end{lemma}
\begin{proof}
Follows directly from the corresponding results in \cite{StojnicICASSP09,StojnicCSetam09,StojnicUpper10}.
\end{proof}
Having matrix $A$ such that
(\ref{eq:thmeqgen}) holds would be enough for solutions of (\ref{eq:l1imphidden}) and (\ref{eq:l0}) (or (\ref{eq:system})) to coincide. As mentioned in the previous section, if one
assumes that $m$ and $k$ are proportional to $n$ (the case of our
interest in this paper) then the construction of the deterministic
matrices $A$ that would satisfy
(\ref{eq:thmeqgen}) is not an easy task (in fact, one may say that together with the ones that correspond to the standard $\ell_1$ it is one of the most fundamental open problems in the area of theoretical compressed sensing). However, to simplify things we will again turn to random matrices.  In fact to be a bit more specific, we will again assume that the elements of matrix $A$ are i.i.d. standard normal random variables. As in the previous section, such an assumption is not really much of a restriction when it comes to generality of the presented results (as mentioned earlier, we will briefly revisit this in Section \ref{sec:conc}). However, to make the presentation as smooth as possible we assume the standard Gaussian scenario.

We then follow the strategy of the previous section (and ultimately the one from \cite{StojnicCSetam09}). To do so, we will make use of Theorem \ref{thm:Gordonmesh}. We start by defining a set $\Swh$
\begin{equation}
\Swh=\{\w\in S^{n-1}| \quad \sum_{i=n-k+1}^{n-\eta k} \w_i<\sum_{i=1}^{n-k-(1-\eta)k}|\w_{i}|\},\label{eq:defSwh}
\end{equation}
where $S^{n-1}$ is the unit sphere in $R^n$. Following what was done in the previous section then effectively means that if one can handle $w(\Swh)$ then, when $n$ is large one can, roughly speaking, use the condition $w(\Swh)< \sqrt{m}$ to obtain an attainable lower bound $\beta_{w}^{(h)}$ for any given $0<\alpha\leq 1$.

To that end we then look at
\begin{equation}
w(S_{w}^{(p)})=E\max_{\w\in S_{w}^{(p)}} (\h^T\w). \label{eq:widthdefswh}
\end{equation}
Following further what was done in Section \ref{sec:analknownsupp} (and earlier in \cite{StojnicCSetam09,StojnicICASSP10knownsupp}) one then can write
\begin{eqnarray}
w(\Swh) & = & E\max_{\w\in \Swh} (\h^T\w)=E\max_{\w\in \Swh} (\sum_{i=1}^{n-k} |\h_i\w_i|+\sum_{i=n-k+1}^{n}\h_i\w_i) \nonumber \\
& = & E\max_{\w\in \Swh} (\sum_{i=1}^{n-k} |\h_i||\w_i|+\sum_{i=n-k+1}^{n}\h_i\w_i).\label{eq:workww0hidden}
\end{eqnarray}
Let $\h_{1:(n-k-(1-\eta)k)}=(\h_1,\h_2,\dots,\h_{n-k-(1-\eta)k})^T$. Further, let now
$|\h|_{(i)}^{(n-k-(1-\eta)k)}$ be the $i$-th smallest magnitude of elements of $\h_{1:(n-k-(1-\eta)k)}$. Set
\begin{equation*}
\hwh=(|\h|_{(1)}^{(n-k-(1-\eta)k)},|\h|_{(2)}^{(n-k-(1-\eta)k)},\dots,|\h|_{(n-k-(1-\eta)k)}^{(n-k-(1-\eta)k)},
\h_{n-k-(1-\eta)k+1},\h_{n-k-(1-\eta)k+2},\dots,\h_{n})^T.
\end{equation*}
Then one can simplify (\ref{eq:workww0hidden}) in the following way
\begin{eqnarray}
w(\Swh) = E\max_{\t\in R^n} & &  \hw^T\t \nonumber \\
\mbox{subject to} &  & \t_i\geq 0, 0\leq i\leq (n-k)\nonumber \\
& & \sum_{i=n-k+1}^{n-\eta k}\t_i\geq \sum_{i=1}^{n-k-(1-\eta)k} \t_i \nonumber \\
& & \sum_{i=1}^n\t_i^2\leq 1.\label{eq:workww1hidden}
\end{eqnarray}
Let $\z\in R^n$ be a column vector such that $\z_i=1,1\leq i\leq (n-k)$, $\z_i=-1,n-k+1\leq i\leq n-\eta k$, and $\z_i=0,n-k-(1-\eta) k+1\leq i\leq n-k$ and $n-\eta k+1\leq i\leq n$. Following step by step the derivation in \cite{StojnicCSetam09} one has, based on the Lagrange duality theory, that there is a $\cweak=(1-\theta_w)n\leq (n-k)$ such that
\begin{multline}
\lim_{n\rightarrow \infty} \frac{w(\Swh)}{\sqrt{n}}=\lim_{n\rightarrow \infty} \frac{E\max_{\w\in S_{w}^{(p)}} (\h^T\w)}{\sqrt{n}}\approxeq
\sqrt{\lim_{n\rightarrow \infty} \frac{E\sum_{i=\cweak+1}^n\hwh_i^2}{n}-\frac{(\lim_{n\rightarrow \infty} \frac{E(\hwh^T\z)-E\sum_{i=1}^{\cweak}\hwh_i}{n})^2}{1-\lim_{n\rightarrow \infty}\frac{\cweak}{n}-\eta \lim_{n\rightarrow \infty}\frac{k}{n}}}.\\\label{eq:wwphidden}
\end{multline}
where $\hwh_i$ is the $i$-th element of vector $\hwh$. Moreover, one has the following counterpart to (\ref{eq:condcwp}) which also establishes the way to determine a critical $\cweak$. Roughly speaking one has the following identity
\begin{equation}
\frac{(\lim_{n\rightarrow \infty} \frac{E(\hwh^T\z)-E\sum_{i=1}^{\cweak}\hwh_i}{n})^2}{1-\lim_{n\rightarrow \infty}\frac{\cweak}{n}-\eta \lim_{n\rightarrow \infty}\frac{k}{n}}\approxeq \lim_{n\rightarrow \infty} \frac{E\hwh_{\cweak}}{n}.\label{eq:condcwphidden}
\end{equation}
Using further the technique of \cite{StojnicCSetam09} one can actually explicitly characterize $w(\Swh)$ in (\ref{eq:wwphidden}) in the following way:
\begin{multline}
\hspace{-.5in}\left ( \lim_{n\rightarrow \infty} \frac{w(\Swh)}{\sqrt{n}} \right )^{2}=\left ( \lim_{n\rightarrow \infty} \frac{E\max_{\w\in S_{w}^{(p)}} (\h^T\w)}{\sqrt{n}} \right )^{2}\\\approxeq
 \frac{1-2\beta_{w}^{(h)}+\eta\beta_{w}^{(h)}}{\sqrt{2\pi}}\left (\sqrt{2\pi}+2\frac{\sqrt{2(\erfinv(\frac{1-\htheta_w-(1-\eta)\beta_{w}^{(h)}}{1-2\beta_{w}^{(h)}+\eta\beta_{w}^{(h)}}))^2}}
 {e^{(\erfinv(\frac{1-\htheta_w-(1-\eta)\beta_{w}^{(h)}}
 {1-2\beta_{w}^{(h)}+\eta\beta_{w}^{(h)}}))^2}}-\sqrt{2\pi}
\frac{1-\htheta_w-(1-\eta)\beta_{w}^{(h)}}{1-2\beta_{w}^{(h)}+\eta\beta_{w}^{(h)}}\right )\\+2\beta_{w}^{(h)}-\eta\beta_{w}^{(h)}
-\frac{\left ((1-2\beta_{w}^{(h)}+\eta\beta_{w}^{(h)})\sqrt{\frac{2}{\pi}}
e^{-(\erfinv(\frac{1-\hat{\theta}_w-(1-\eta)\beta_{w}^{(h)}}{1-2\beta_{w}^{(h)}+\eta\beta_{w}^{(h)}}))^2}\right )^2}{\hat{\theta}_w-\eta\beta_{w}^{(h)}},\label{eq:wwp1hidden}
\end{multline}
where $\htheta_w$ is the solution of
\begin{equation}
\frac{\left ((1-2\beta_{w}^{(h)}+\eta\beta_{w}^{(h)})\sqrt{\frac{2}{\pi}}e^{-(\erfinv(\frac{1-\theta_w-(1-\eta)\beta_{w}^{(h)}}{1-2\beta_{w}^{(h)}+\eta\beta_{w}^{(h)}}))^2}\right )}{\theta_w-\eta\beta_{w}^{(h)}}\approxeq
\sqrt{2}\mbox{erfinv}\left (\frac{1-\theta_w-(1-\eta)\beta_{w}^{(h)}}{1-2\beta_{w}^{(h)}+\eta\beta_{w}^{(h)}}\right ).\label{eq:condcwp1}
\end{equation}

We summarize the above results in the following theorem.

\begin{theorem}(Location of fraction $\eta$ of non-zero elements is known to be hidden with a set $\kappa$)
Let $A$ be an $m\times n$ measurement matrix in (\ref{eq:system})
with the null-space uniformly distributed in the Grassmanian. Let
the unknown $\x$ in (\ref{eq:system}) be $k$-sparse. Further, let the location and signs of nonzero elements of $\x$ be arbitrarily chosen but fixed. Moreover, let the set of nonzero locations of $\x$ be $K$. Let $\kappa\subset \{1,2,\dots,n\}$ be a given set of cardinality $k$ such that the cardinality of set $K\cap \kappa$ is $\eta k$. Let $k,m,n$ be large
and let $\alpha=\frac{m}{n}$, $\beta_w^{(h)}=\frac{k}{n}$, and $\eta$ be constants
independent of $m$, $n$, and $k$. Let $\erfinv$ be the inverse of the standard error function associated with zero-mean unit variance Gaussian random variable.  Further,
let $\htheta_w$, ($\beta_w^{(h)}\leq \htheta_w\leq 1-\beta_w^{(h)}+\eta\beta_w^{(h)}$) be the solution of
\begin{equation}
(1-2\beta_w^{(h)}+\eta\beta_w^{(h)})\frac{\sqrt{\frac{2}{\pi}}e^{-(\erfinv(\frac{1-\theta_w-(1-\eta)
\beta_w^{(h)}}{1-2\beta_w+\eta\beta_w^{(h)}}))^2}}{\theta_w-\eta\beta_w^{(h)}}\approxeq\sqrt{2}\erfinv (\frac{1-\theta_w-(1-\eta)\beta_w^{(h)}}{1-2\beta_w^{(h)}+\eta\beta_w^{(h)}})=0.\label{eq:thmweakthetahidden}
\end{equation}

\noindent If $\alpha$ and $\beta_w^{(h)}$ further satisfy
\begin{multline}
\alpha>\frac{1-2\beta_w^{(h)}+\eta\beta_w^{(h)}}{\sqrt{2\pi}}\left (\sqrt{2\pi}+2\frac{\sqrt{2(\erfinv(\frac{1-\htheta_w-(1-\eta)\beta_w^{(h)}}{1-2\beta_w^{(h)}+\eta\beta_w^{(h)}}))^2}}{e^{(\erfinv(\frac{1-\htheta_w-(1-\eta)\beta_w^{(h)}}
{1-2\beta_w^{(h)}+\eta\beta_w^{(h)}}))^2}}
-\sqrt{2\pi}
\frac{1-\htheta_w-(1-\eta)\beta_w^{(h)}}{1-2\beta_w^{(h)}+\eta\beta_w^{(h)}}\right )\\+2\beta_w^{(h)}-\eta\beta_w^{(h)}
-\frac{\left ((1-2\beta_w^{(h)}+\eta\beta_w^{(h)})\sqrt{\frac{2}{\pi}}e^{-(\erfinv(\frac{1-\hat{\theta}_w-(1-\eta)\beta_w^{(h)}}{1-2\beta_w^{(h)}+\eta\beta_w^{(h)}}))^2}\right )^2}{\hat{\theta}_w-\eta\beta_w^{(h)}}\label{eq:thmweakalphahidden}
\end{multline}
then the solution of (\ref{eq:l1imphidden}) and the $k$-sparse solution of (\ref{eq:system}) coincide with overwhelming
probability.

\noindent If $\alpha$ and $\beta_w^{(h)}$ further satisfy
\begin{multline}
\alpha<\frac{1-2\beta_w^{(h)}+\eta\beta_w^{(h)}}{\sqrt{2\pi}}\left (\sqrt{2\pi}+2\frac{\sqrt{2(\erfinv(\frac{1-\htheta_w-(1-\eta)\beta_w^{(h)}}{1-2\beta_w^{(h)}+\eta\beta_w^{(h)}}))^2}}{e^{(\erfinv(\frac{1-\htheta_w-(1-\eta)\beta_w^{(h)}}
{1-2\beta_w^{(h)}+\eta\beta_w^{(h)}}))^2}}
-\sqrt{2\pi}
\frac{1-\htheta_w-(1-\eta)\beta_w^{(h)}}{1-2\beta_w^{(h)}+\eta\beta_w^{(h)}}\right )\\+2\beta_w^{(h)}-\eta\beta_w^{(h)}
-\frac{\left ((1-2\beta_w^{(h)}+\eta\beta_w^{(h)})\sqrt{\frac{2}{\pi}}e^{-(\erfinv(\frac{1-\hat{\theta}_w-(1-\eta)\beta_w^{(h)}}{1-2\beta_w^{(h)}+\eta\beta_w^{(h)}}))^2}\right )^2}{\hat{\theta}_w-\eta\beta_w^{(h)}}\label{eq:thmweakalphahidden1}
\end{multline}
then with overwhelming probability there will be a $k$-sparse $\x$ (from a set of $\x$'s with fixed locations and signs of nonzero components) that satisfies (\ref{eq:system}) and is \textbf{not} the solution of (\ref{eq:l1imphidden}).
\label{thm:thmweakthrknownsupphidden}
\end{theorem}
\begin{proof}
The first part follows from the above discussion as well as from the considerations presented in Section \ref{sec:analknownsupp} and earlier in \cite{StojnicICASSP10knownsupp,StojnicCSetam09}. The moreover part follows by a combination of the moreover part of Lemma \ref{lemma:lemmaknownsuppcondhidden} and the considerations presented in \cite{StojnicUpper10,StojnicGorEx10}.
\end{proof}

\noindent \textbf{Remark:} As in Section \ref{sec:knownsupp} to make writing easier in the previous theorem we removed all $\epsilon$'s used in Theorem \ref{thm:thmweakthr}.

In a more informal language one then has the following interpretation of the above theorem. Assume the setup of the above theorem. Let $\alpha_w^{(h)}$ and $\beta_{w}^{(h)}$ satisfy the following:

\noindent \underline{\underline{\textbf{Fundamental characterization of the \emph{hidden} partial $\ell_1$ performance:}}}

\begin{center}
\shadowbox{$
(1-2\beta_{w}^{(h)}+\eta\beta_{w}^{(h)})\frac{\sqrt{\frac{2}{\pi}}e^{-(\erfinv(\frac{1-\alpha_w^{(h)}}{1-2\beta_{w}^{(h)}+\eta\beta_{w}^{(h)}}))^2}}
{\alpha_w^{(h)}-\beta_{w}^{(h)}}-\sqrt{2}\erfinv (\frac{1-\alpha_w^{(h)}}{1-2\beta_{w}^{(h)}+\eta\beta_{w}^{(h)}})=0.
$}
-\vspace{-.5in}\begin{equation}
\label{eq:thmweakinfknownsupp}
\end{equation}
\end{center}

Then:
\begin{enumerate}
\item If $\alpha>\alpha_w^{(h)}$ then with overwhelming probability the solution of (\ref{eq:l1imphidden}) is the $k$-sparse $\x$ from (\ref{eq:system}).
\item If $\alpha<\alpha_w^{(h)}$ then with overwhelming probability there will be a $k$-sparse $\x$ (from a set of $\x$'s with fixed locations and signs of nonzero components) that satisfies (\ref{eq:system}) and is \textbf{not} the solution of (\ref{eq:l1imphidden}).
    \end{enumerate}

The above theorem essentially settles typical behavior of the \emph{hidden} partial $\ell_1$ optimization from (\ref{eq:l1imphidden}) when used for solving (\ref{eq:system}) or (\ref{eq:l0}) assuming that a fraction $\eta$ of nonzero locations of $\x$ is a priori known to be within a set $\kappa$of cardinality $k$.

The results for the weak threshold obtained from the above theorem
are presented in Figure \ref{fig:weakknownsupphidden}. As $\eta$ increases more knowledge about $\x$ is available and one expects that the threshold values of the recoverable sparsity should be higher. As results presented in Figure \ref{fig:weakknownsupphidden} indicate, the values of the threshold recoverable by the modified hidden partial $\ell_1$ optimization from (\ref{eq:l1imphidden}) are indeed higher as $\eta$ increases.

Carefully looking at the results presented in Figure \ref{fig:weakknownsupphidden} one can note that performing optimization from (\ref{eq:l1imphidden}) is beneficial when compared to $\ell_1$ only in the regime above the dashed blue curve. Also for $\eta=0.75$ the improvement seems quite substantial. For example, if one can locate a set $\kappa$ of cardinality $k$ that contains $75\%$ of the support of $\x$ then one could provably substantially improve performance of $\ell_1$ from (\ref{eq:system}). What is even more interesting is that many of algorithms designed to solve (\ref{eq:system}) usually work in such a way that even when they fail to recover the entire support they still recover correctly significant portion of the support. Where these algorithms usually fail is inability to locate where, within the offered incorrect estimate of the entire support, that portion is. The results presented in Figures \ref{fig:weakknownsupp} and \ref{fig:weakknownsupphidden} provide a solid intuitive justification as to why such algorithms may still have a chance to outperform $\ell_1$.

In addition to the theoretical results one can obtain using Theorem \ref{thm:thmweakthrknownsupphidden} we on the right side of Figure \ref{fig:weakknownsupp} show experimental results one can obtain through numerical simulations. We discuss these in a bit more detail below.
\begin{figure}[htb]
\begin{minipage}[b]{0.5\linewidth}
\centering
\centerline{\epsfig{figure=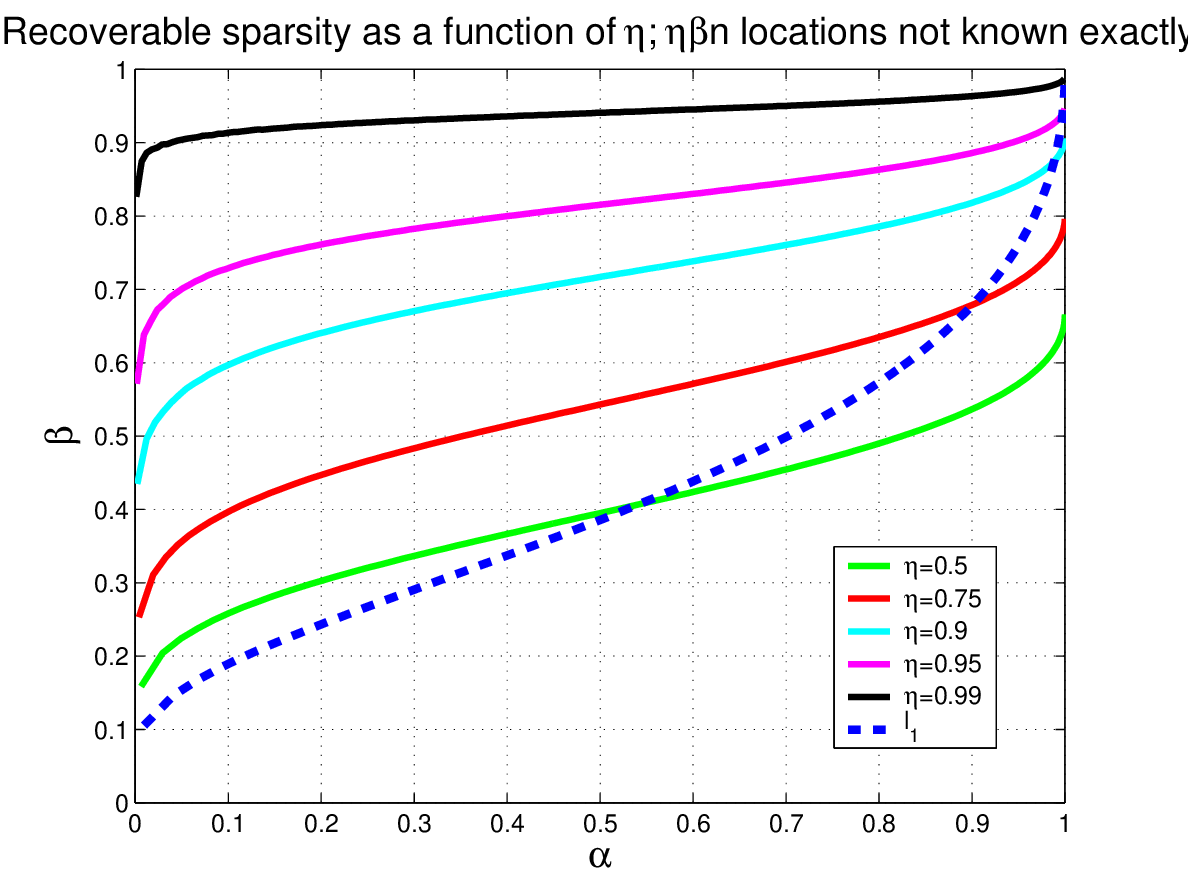,width=8cm,height=6cm}}
\end{minipage}
\begin{minipage}[b]{0.5\linewidth}
\centering
\centerline{\epsfig{figure=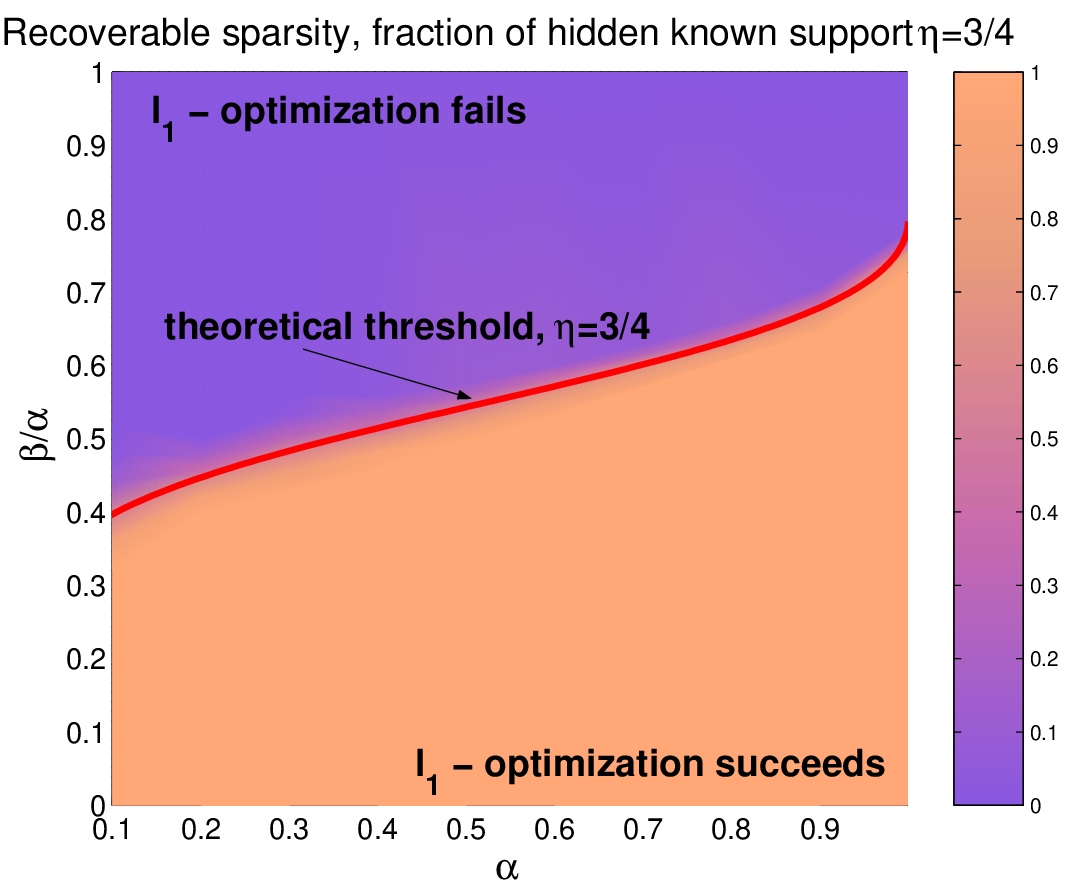,width=8cm,height=6cm}}
\end{minipage}
\caption{Left: Theoretical \emph{weak} threshold as a function of fraction of \emph{hidden} known support; Right: Experimentally recoverable sparsity; fraction of \emph{hidden} known support $\eta=\frac{3}{4}$}
\label{fig:weakknownsupphidden}
\end{figure}

\subsection{Numerical experiments}
\label{sec:simulh}

In this section we briefly discuss the results that we obtained from numerical experiments. In our numerical experiments we selected $n=1000$ when $\alpha\leq 0.2$ (to obtain a finer resolution) and $n=500$ when $\alpha>0.2$. In all experiments we assumed $\eta=0.75$. We then generated matrices $A$ of size $m\times n$ with $m=(0.1n,0.2n,0.3n,\dots,0.9n,0.99n)$. The components of the measurement matrices $A$ were generated as i.i.d. zero-mean unit variance Gaussian random variables.
As in Section \ref{sec:knownsupp}, for each $m$ we generated $k$-sparse signals $\x$ for several different values of $k$ from the transition zone (the locations of non-zero elements of $\x$ were chosen randomly and three quarters or $\eta$ of them that correspond to the known part, i.e. to $\Pi$, were chosen randomly as well; the remaining components of $\kappa$ were chosen randomly as well). For each combination $(k,m)$ we generated $100$ different problem instances and recorded the number of times the hidden partial $\ell_1$-optimization algorithm from (\ref{eq:l1imphidden}) failed to recover the correct $k$-sparse $\x$.
The obtained data are then interpolated and graphically presented on the right hand side of Figure \ref{fig:weakknownsupphidden}. As in Section \ref{sec:knownsupp}, the color of any point shows the probability of having the hidden partial $\ell_1$-optimization from (\ref{eq:l1imphidden}) succeed for a combination $(\alpha,\beta)$ that corresponds to that point. The colors are again mapped to probabilities according to the scale on the right hand side of the figure.
The simulated results can naturally be compared to the theoretical prediction from Theorem \ref{thm:thmweakthrknownsupphidden}. Hence, we also on the right hand side plot the theoretical value for the threshold calculated according to Theorem \ref{thm:thmweakthrknownsupphidden} (obviously these are shown on the left hand side of the figure as well). We again observe that the simulation results are in a good agreement with the theoretical calculation.

\section{Discussion}
\label{sec:discussion}

In this section we briefly look at the presented results and how they fit into a larger scope, especially within the framework presented in \cite{StojnicReDirChall13}. Namely, in \cite{StojnicReDirChall13}, we observed that since the original work of Donoho \cite{DonohoUnsigned,DonohoPol} appeared almost $10$ years ago not much changed in the location of the achievable recoverable thresholds. Of course we quickly pointed out that not much can be changed since Donoho actually determined the performance characterization of the $\ell_1$-optimization. However, what was really emphasized in \cite{StojnicReDirChall13} is that there has not been alternative characterizations that go above the one presented for $\eta=0$ in Figure \ref{fig:weakknownsupp} (of course assuming that they are obtained through an analysis of a polynomial algorithm). Now, looking at what we presented in this paper (for example in the very same Figure \ref{fig:weakknownsupp} or alternatively in Figure \ref{fig:weakknownsupphidden}) as well as what was presented in many other papers either experimentally or theoretically (see, e.g. \cite{StojnicICASSP10knownsupp,CWBreweighted,SChretien08,SaZh08,StojnicCSetamBlock09,StojnicUpperBlock10,StojnicISIT2010binary}), one may be tempted to object such a statement. The reason of course could be that some of the plots in Figures \ref{fig:weakknownsupp} and \ref{fig:weakknownsupphidden} are higher than the curve $\eta=0$ in Figure \ref{fig:weakknownsupp}. However, as mentioned in \cite{StojnicReDirChall13}, while there are scenarios where the characterizations can be lifted, it is not clear to us if one could consider any of such lifts as a ``universal" lift of the characterization $\eta=0$ in Figure \ref{fig:weakknownsupp}. While more on our understanding of a ``universal" lift can be found in \cite{StojnicReDirChall13}, here we just briefly recall on a question we posed in \cite{StojnicReDirChall13} relying on such an understanding:

\textbf{Question 1}:
Let $A$ be an $\alpha n\times n$ matrix with i.i.d standard normal components. Let $\tilde{\x}$ be a $\beta n$ -sparse $n$-dimensional vector from $R^n$ and let the signs and locations of its non-zero components be arbitrarily chosen but fixed. Moreover, let pair $(\beta,\alpha)$ reside in the area above the curve given for $\eta=0$ in Figure \ref{fig:weakknownsupp}. Can one then design a polynomial algorithm that would with overwhelming probability (taken over randomness of $A$) solve (\ref{eq:system}) for all such $\tilde{\x}$?

Without going into the details about possible deficiencies in the formulation of the above question (these are to a large extent discussed in \cite{StojnicReDirChall13}) we here only briefly discuss what kind of consequences the results presented in this paper have on it. First, as already mentioned above, just by looking at plots in Figures \ref{fig:weakknownsupp} and \ref{fig:weakknownsupphidden} one immediately may wonder isn't the answer to Question $1$ yes. The fact that all of the curves presented in Figure \ref{fig:weakknownsupp} and some of the curves presented in Figure \ref{fig:weakknownsupphidden} are indeed well above $\eta=0$ curve is actually not enough to conclude that the answer to the above question is yes. One should keep in mind that whenever $\eta>0$ one essentially uses an extra amount of knowledge about $\x$ which is a luxury that the original problem (\ref{eq:l0}) does not have. Still, the results presented in Figures \ref{fig:weakknownsupp} and \ref{fig:weakknownsupphidden} provide in a way the following useful information: namely, if one can determine in some way a certain fraction of $supp(\x)$ then the answer to Question $1$ could be yes. One has to be careful though, because such a set of locations has to be pretty much random with respect to the original $supp(\x)$. This is of course the key obstacle why the results presented in Sections \ref{sec:knownsupp} and \ref{sec:hidknownsupp} may not be enough to resolve Question $1$ in positive. What typically happens when any of the iterative algorithms are employed is roughly the following: one can often correctly guess a fairly large fraction of $supp(\x)$ even in the range $(\alpha,\beta)$ above the fundamental $\ell_1$ performance characterization. The problem is that such a guess is almost always in a way biased, i.e. it does not contain a random fraction of $supp(\x)$ but rather a carefully selected fraction of $supp(\x)$. That of course is not enough to utilize the machinery presented in this paper. Still, we believe that the results we presented in this paper in a way simplify what is sufficient to be done if one is to resolve Question $1$ in positive. Of course, if the answer to Question $1$ is no then such a simplification may not be of very much use.

We should also mention that the questions we posed in \cite{StojnicReDirChall13} are purely mathematical. If they could be resolved in positive then they would have significant practical implications as well. However if they can not, we believe that their importance is purely on a theoretical level. On the other hand if one ignores the mathematical frame from \cite{StojnicReDirChall13} and views the results presented in this paper in a practical context then their value seems quite significant. For example, in many practical situations one may be able to have an available feedback about $supp(\x)$. While availability of such a feedback precludes a fair comparison between the curves in Figures \ref{fig:weakknownsupp} and \ref{fig:weakknownsupphidden} (and consequently their direct use in providing any definite answer to Question $1$) the results presented in this paper are very useful as they characterize performances of relatively simple $\ell_1$ modifications given in (\ref{eq:l1imp}) and (\ref{eq:l1imphidden}).

\section{Conclusion}
\label{sec:conc}

In this paper we looked at possible modifications of standard $\ell_1$ optimization when used for recovering sparse solutions of under-determined systems of linear equations. More specifically, we considered two modifications that can be used in scenarios when some kind of information about the support of unknown vector is a priori available. First, we considered scenario which assumes that a fraction of the support of unknown vector is known and then we looked at the scenario which assumes that a given set of locations contains a fraction of the support of unknown vectors. For both of these modifications, in a statistical context we provided a precise characterization of systems dimensions for which they successfully find the sparsest solution of the system.

As was the case in \cite{StojnicCSetam09,StojnicMoreSophHopBnds10,StojnicLiftStrSec13}, the purely theoretical results we presented in this paper are valid for the so-called Gaussian models, i.e. for systems with i.i.d. Gaussian coefficients. Such an assumption significantly simplified our exposition. However, all results that we presented can easily be extended to the case of many other models of randomness. There are many ways how this can be done. Instead of recalling on them here we refer to a brief discussion about it that we presented in \cite{StojnicMoreSophHopBnds10}.

As for usefulness of the presented results, there is hardly any limit. One can look at a host of related problems from the compressed sensing literature. These include for example, all noisy variations, approximately sparse unknown vectors, vectors with a priori known structure (block-sparse, binary/box constrained etc.), all types of low rank matrix recoveries, various other algorithms like $\ell_q$-optimization, SOCP's, LASSO's, and many, many others. Each of these problems has its own specificities and adapting the methodology presented here usually takes a bit of work but in our view is now a routine. While we will present some of these applications we should emphasize that their contribution will be purely on an application level.

\begin{singlespace}
\bibliographystyle{plain}
\bibliography{TowBettCompSensRefs}
\end{singlespace}

\end{document}